\RequirePackage{fix-cm}
\documentclass[smallextended]{svjour3}       
\smartqed  
\usepackage{graphicx}
\usepackage{bm}
\usepackage[mathlines]{lineno}
\usepackage{xcolor}
\usepackage{color}
\usepackage{soul}
\usepackage[T1]{fontenc}
\usepackage{amsfonts}

\usepackage[title]{appendix}
\usepackage{esvect}
\usepackage{amsmath}

\usepackage[framemethod=tikz]{mdframed}

\usepackage{physics}
\begin{document}

\title{Continuous Time Limit of the DTQW in 2D+1 and Plasticity 
}


\author{Michael Manighalam         \and
        Giuseppe Di Molfetta 
}


\institute{Michael Manighalam \at
              Boston University \\
              Boston, MA\\
              \email{mbmanigh@bu.edu}           
           \and
           Giuseppe Di Molfetta \at
              Aix-Marseille Universit{\'e},\\
              Universit\'e de Toulon,\\
              CNRS, LIS, Marseille, France\\
              \email{giuseppe.dimolfetta@lis-lab.fr} 
}

\date{Received: date / Accepted: date}

\maketitle
%
%
%
%
%
%
\begin{abstract} 
A Plastic Quantum Walk admits both continuous time and continuous spacetime. The model has been recently proposed by one of the authors in~\cite{molfetta2019quantum}, leading to a general quantum simulation scheme for simulating fermions in the relativistic and non relativistic regimes. 
The extension to two physical dimensions is still missing and here, as a novel result, we demonstrate necessary and sufficient conditions concerning which discrete time quantum walks can admit plasticity, showing the resulting Hamiltonians. 
We consider coin operators as general $4$ parameter unitary matrices, with parameters which are functions of the lattice step size $\varepsilon$. This dependence on $\varepsilon$ encapsulates all functions of $\varepsilon$ for which a Taylor series expansion in $\varepsilon$ is well defined, making our results very general.  

\end{abstract}

\newpage
%
%

\section{Introduction}
Confronted with the insufficiency and intractability of classical computers' abilities to simulating quantum systems, the idea of simulating quantum systems with quantum computers was born. Such inefficiencies with classic computers were notably pointed out by Feynman~\cite{FeynmanQC}, which sparked quantum simulation schemes to be the subject of much attention over the last few decades~\cite{georgescu2014quantum}. Some of the methods being used for simulating quantum systems implemented over discrete space continuous time lattices consist of constructing a Hamiltonian which imitates a physical system, or trotterizing a constructed Hamiltonian to obtain unitaries~\cite{jordan2012quantum,StrauchCTQW}. Problems with these approaches are discussed in Ref.~\cite{molfetta2019quantum} and include the breaking of Lorentz covariance as well as issues arising when recovering a bounded speed of light. Discrete spacetime models are also used to simulate quantum systems which do not share the difficulties of their discrete space continuous time counterparts, such as the quantum circuit model and the discrete time quantum walk (DTQW), the latter being the focus of this work. Concerning the simulation of quantum systems by DTQWs, it has been discussed in Ref.~\cite{Arnault_2017} that the continuous spacetime limit of various DTQWs defined on the regular lattice in arbitrary dimensions is equivalent to coupled Dirac Fermion dynamics with abelian \cite{di2012discrete,MolfettaDebbasch2014Curved,Arnault_2016} and non-abelian gauge field \cite{di2016quantum,arnault2016quantum,ARNAULT2016179} on curved spacetime \cite{di2013quantum,ArrighiGRDirac3D,succiQWBoltzmann,arrighi2019curved}. Concerning the DTQWs ability to simulate discrete space-continuous time quantum systems, it has been shown that the continuous time limit of the DTQW coincides to the continuous time quantum walk (CTQW), which is equivalent to the finite-difference Schrodinger's equation \cite{StrauchCTQW}. Also, recently a quantum simulation scheme known as a Plastic Quantum Walk has been developed which supports both a continuous spacetime limit and a continuous time-discrete space limit, and the procedure for obtaining such a walk yields a curved spacetime Hamiltonian for lattice-fermions with synchronous coordinates~\cite{molfetta2019quantum}.

While this has opened the route for elaborating universal QW based simulators of interacting particles in relativistic ($\Delta_x \ll 1$) and non-relativistic regime ($\Delta_x = 1$), a generalisation to higher dimensional spacetime is missing from that work. In this analysis we aim to do just that: introduce a novel and very general method of computing a Plastic DTQW in 2D+1, where we maintain the "spirit" of quantum walks as much as possible (i.e. we constrain coin parameters to not depend on time or lattice position). This generality is obtained by maintaining as many tunable parameters in our coin operators as possible through the continuum limit and minimally constraining coin parameters to be any functions of the lattice step size  $\varepsilon$ for which a Taylor series exists. As in ~\cite{molfetta2019quantum} the necessity of an even stroboscopic step size for 2D+1 continuous time limits of DTQWs is also recovered, which is an original result. Lastly, the continuous equations we obtain are original and very general as well and recover the lattice fermion Hamiltonian in continuous time and the Dirac equation in continuous spacetime in 2D with an opportune particular choice of the parameters. This will lead to an alternative operational formal model useful to the development of quantum simulators of gauge invariant models on the grid, in particular the Kogut-Susskind Hamiltonian \cite{kogut1975hamiltonian} completely alternative to the standard formulation of lattice gauge theories \cite{zohar2015formulation}.\\

\textit{Roadmap} Section \ref{sec:model} presents the QW. Section \ref{sec:limits} shows the different scalings and continuous time limits it supports. Then in Section \ref{sec:plasticQW} we consider the subset of sufficient and necessary conditions which allows the QW, called Plastic, to admit both a continuous time and a continuous spacetime limit. Finally Section \ref{sec:conclusion} summarizes the results, and concludes. We added three appendices for detailed proofs.

\section{Model \label{sec:model}}

We consider a QW over the 2D+1--spacetime grid. Its coin or spin degree of freedom lies $\mathcal{H}_2$, for which we may chose some orthonormal basis $\{\ket{v^L}, \ket{v^R}\}$. The overall state of the walker lies in the composite Hilbert space $\mathcal{H}_2\otimes \mathcal{H}_\mathbb{Z}^2$ and may be thus be written $\Psi=\sum_{l,m} \psi^L(l,m) \ket{v_L}\otimes\ket{l,m} + \psi^R(l,m) \ket{v_R}\otimes\ket{l,m}$, where the scalar field $\psi^L$ (resp. $\psi^R$) gives the amplitude of the particle being there and about to move left (resp. right) at every position $(l,m)\in \mathbb{Z}^2$. We use $(n,l,m) \in \mathbb{N} \times \mathbb{Z}^2$ to label instants and points in space, respectively, and let:
\begin{equation}\label{eq:time_evolution}
\Psi_{n+1}=W \Psi_n
\end{equation}
where 
\begin{equation}
W = V_x V_y
\end{equation}
and
\begin{equation}
V_i = S_i(C_i\otimes \text{Id}_\mathbb{Z}) 
\end{equation}
with
$S_i$ a state-dependent shift operator such that 
\begin{equation}
(S_x\Psi)_{n,l,m}  =\begin{pmatrix}\psi^L_{n,l+1,m}\\\psi^R_{n,l-1,m}\end{pmatrix}
\end{equation} 
and
\begin{equation}
(S_y\Psi)_{n,l,m}  =\begin{pmatrix}\psi^L_{n,l,m+1}\\\psi^R_{n,l,m-1}\end{pmatrix}
\end{equation} 
and $C_x$ and $C_y$ are elements of $U(2)$ and depend on the four real parameters $\delta_j$, $\zeta_j$, $\theta_j$, and $\phi_j$ in the following way (where $j=x$ or $y$):
\begin{equation}\label{eq:rotations}
    \begin{split}
        C_j&=e^{i\delta_j}R_z(\zeta_j)R_y(\theta_j)R_z(\phi_j)=e^{i\delta_j}e^{ -i\zeta_j\sigma_z/2}e^{ -i\theta_j\sigma_y/2} e^{ -i\phi_j\sigma_z/2}\\
        &=e^{i\delta_j}\begin{pmatrix}\cos\frac{\theta_j}{2}\exp-i\frac{\phi_j+\zeta_j}{2}&-\sin\frac{\theta_j}{2}\exp i\frac{\phi_j-\zeta_j}{2}\\\sin\frac{\theta_j}{2}\exp i\frac{-\phi_j+\zeta_j}{2}&\cos\frac{\theta_j}{2}\exp i\frac{\phi_j+\zeta_j}{2}\end{pmatrix}
    \end{split}
\end{equation}
To investigate the continuous limits, we first introduce a time discretization step $\Delta_t$ and a space discretization step $\Delta$ for both the $x$ and $y$ dimension. We then introduce, for any discrete function $\Psi$ appearing in Eq.~\eqref{eq:time_evolution}, a field $\tilde \Psi$ over the spacetime positions $\mathbb{R}^+ \times \mathbb{R}$, such that $\Psi_{n,l,m}=\tilde \Psi(t_n,x_l,y_m)$, with $t_n=n \Delta_t$, $x_l = l \Delta$, and $y_m = m \Delta$. 
Eq. \eqref{eq:time_evolution} then reads: 
\begin{equation}
\widetilde{\Psi}(t_n+\Delta_t) =  W \widetilde{\Psi}(t_n).
\label{eq:time_evolution2}
\end{equation}

Let us drop the tildes to lighten the notation. We suppose that all functions are $C^2$. In general the spacetime continuum limit, when it exists, is the coupled differential equations obtained from Eq. \eqref{eq:time_evolution2} by letting both $\Delta_t$ and $\Delta$ go to zero, as for example in \cite{di2012discrete,Arnault_2017}. When we are interested in choosing to let one of them go to zero, for instance $\Delta_t$, the result is a lattice Hamiltonian equation. If the above walk admits both limit, we will call it Plastic. \\

In the following section we will investigate first the necessary and sufficient conditions for the continuous time limit, as, usually is a sub-set of those to recover the continuous spacetime limit.

\subsection{Continuous time limit} 

In the following, we will find for which parameters $\delta_j$, $\zeta_j$, $\theta_j$, $\phi_j$, and $\tau$ (the stroboscopic step size) the continuous time limit of Eq. \eqref{eq:time_evolution2} exists and converges to:
\begin{equation}\label{eq:real_space_ham}
   H \Psi(t)=i\partial_t\Psi(t)=i\lim\limits_{\Delta_t\to 0} \frac{W^\tau-\mathbb{I}}{\tau\Delta_t}\Psi(t).
\end{equation}
In particular, $\Delta$ remains finite and without loss of generality we can normalise it to unity. To prove our main result we represent our walk in Fourier space and we define our discrete Fourier transform convention here. Let $\hat{\psi^a}(t,k_x,k_y)$, with $a=\{L,R\}$, be the Fourier transform of $\psi^a(t,x_l,y_m)$.  We use the following conventions for the forward and inverse Fourier transforms, with Fourier variables $(k_x,k_y)\in[-\pi,\pi]^2$: 
\begin{align}
    \hat{\psi^a}(t,k_x,k_y)&=\sum_{l=-\infty}^{\infty}\sum_{m=-\infty}^{\infty}e^{-ik_x l}e^{-ik_y m}\psi^a(t,x_l,y_m)\equiv \mathcal F(\psi^a)\\
    \psi^a(t,x_l,y_m)&=\frac{1}{(2\pi)^2}\int^{\pi}_{-\pi}dk_x\int^{\pi}_{-\pi}dk_ye^{ik_x l}e^{ik_y m}\hat{\psi^a}(t,k_x,k_y)\equiv\mathcal{F}^{-1}(\widehat \psi^a).
\end{align}
A standard procedure is to represent operators in Fourier space as follows: given an operator $O$ on a function space $Y$, its Fourier conjugate operator $\hat O$ is defined by $\hat O \hat f(k)= \mathcal{F}(O(f(x)))$, with $f(x)\in Y$, so that $\hat O$ is the Fourier representation of $O$. In particular, the shift operators $S_x$ and $S_y$ in Fourier space translate:

\begin{equation}\label{eq:defn}
    \begin{split}
        &\mathcal{F}(S_x\Psi(t))=\widehat{S_x}\hat{\Psi}(t)=e^{ik_x\sigma_z}\hat{\Psi}(t)=R_z(-2k_x)\hat{\Psi}(t)\\
        &\mathcal{F}(S_y\Psi(t))=\widehat{S_y}\hat{\Psi}(t)=e^{ik_y\sigma_z}\hat{\Psi}(t)=R_z(-2k_y)\hat{\Psi}(t).
    \end{split}
\end{equation}

The time evolution Eq. \eqref{eq:time_evolution2} in Fourier space then becomes the following:
\begin{equation}
    \begin{split}
        \hat{\Psi}(t+\Delta_t)&=\widehat{W}\hat{\Psi}(t) =e^{ik_x \sigma_z}C_xe^{ik_y \sigma_z}C_y\hat{\Psi}(t)
    \end{split}
\end{equation}
and Eq. \eqref{eq:real_space_ham} reduces to :
\begin{equation}\label{eq:ham}
    \widehat{H}\hat{\Psi}(t)=i\partial_t\hat{\Psi}(t)=i\lim\limits_{\Delta_t \to 0}\frac{(e^{ik_x \sigma_z}C_xe^{ik_y \sigma_z}C_y)^\tau-\mathbb{I}}{\tau\Delta_t}\hat{\Psi}(t).
\end{equation}

%
%
%
%

\section{Continuum limit and scalings}\label{sec:limits}

In order to find the continuum limit in equation \eqref{eq:ham}, let us first parametrize the four real parameters defining the quantum coin, as follows:
\begin{equation}\label{eq:scaling}
    \begin{split}
        \zeta_j&=\zeta_{0j}+\zeta_{1j}\Delta_t\\
        \theta_j&=\theta_{0j}+\theta_{1j}\Delta_t\\
        \phi_j&=\phi_{0j}+\phi_{1j}\Delta_t.
    \end{split}
\end{equation}
Altogether, these jets define a family of QWs indexed by $\Delta_t$, whose embedding in spacetime, and defining angles, depend on $\Delta_t$. The continuum limit of Eq.\eqref{eq:ham} can then be investigated by Taylor expanding $\Psi(t)$ around $(t_n,x_l,y_m)$.

Using Eq.\eqref{eq:scaling}, and expanding around $\Delta_t = 0$, the rotation matrices $R_m(w)$ read:
\begin{equation}
    R_m(w)\simeq R_m(w_0)(1-\frac{iw_1\Delta_t}{2}\sigma_m+O(\Delta_t^2))
\end{equation}
 where $w=\zeta$,$\theta$,$\phi$ and $m=x,y$. We also recover the first order of the split-step unitary operator, leaving the proof to Appendix~\ref{app:Wexpansion}:

\begin{equation}\label{eq:w}
    \begin{split}
    \widehat{W}\simeq
    e^{i\delta}(A-\frac{i \Delta_t}{2} B+O(\Delta_t^2))
    \end{split}
\end{equation}
where $A=A_xA_y$, $B=A_xB_y+B_xA_y$, $\delta = \delta_x+\delta_y$ and 
\begin{equation}
\begin{split}
 A_j&=R_z(\zeta'_{0j})R_y(\theta_{0j})R_z(\phi_{0j})\\
 B_j&=\zeta_{1j}\sigma_z A_j+\theta_{1j} \sigma_y R_z(-2\zeta'_{0j})A_j+\phi_{1j} A_j \sigma_z\\
\zeta'_{0j}&=\zeta_{0j}-2k_j.
\end{split}
\end{equation}
Finally, in order to compute the leading orders of Eq. \eqref{eq:ham} we need to compute the $\tau^{th}-$power of the above operator. The $\tau^{th}-$power of $W$:
\begin{equation}\label{eq:reduced1sc^n}
    \begin{split}
    \widehat{W}^\tau&\simeq e^{i\delta \tau}(A-\frac{i \Delta_t}{2} B)^\tau =(e^{i\delta}A)^\tau(\mathbb{I}-\frac{i \Delta_t}{2}A^{-1}\sum_{j=0}^{\tau-1}A^{-j}BA^j+O(\Delta_t^2)).
    \end{split}
\end{equation}
For detailed proof of Eq.~\eqref{eq:reduced1sc^n}, see Appendix~\ref{app:Wexpansion}.\\

Now we have the following two lemmas:
\begin{lemma}\label{lma:Odt^2}
The continuous time limit as defined in Eq. \eqref{eq:ham} will be independent of any $O(\Delta_t^2)$ terms in the parameters $\zeta$, $\theta$, and $\phi$.
\end{lemma}
\begin{proof}
We see that the only contribution of the $O(\Delta_t^2)$ terms in the parameters $\zeta$, $\theta$, and $\phi$ will be in the $O(\Delta_t^2)$ term. The $O(\Delta_t^2)$ term in Eq.~\eqref{eq:reduced1sc^n} does not contribute to the continuous time limit defined in Eq.~\eqref{eq:ham} because it goes to zero as the limit is taken. Thus, the $O(\Delta_t^2)$ terms in the parameters $\zeta$, $\theta$, and $\phi$ do not contribute to the continuous time limit.\qed
\end{proof}

\begin{lemma}\label{lma:n=1}
There is no continuous time limit as defined in Eq. \eqref{eq:ham} for $\tau=1$.
\end{lemma}
\begin{proof}
For the Hamiltonian in Eq.~\eqref{eq:ham} to be finite, $\widehat{W}^\tau$ must equal $\mathbb{I}+O(\Delta_t)$, and thus $\widehat{W}$ must equal $\mathbb{I}+O(\Delta_t)$ as well. Therefore, from Eq.~\eqref{eq:reduced1sc^n}, $(e^{i\delta}A)^\tau$ must equal identity if $\widehat{W}=\mathbb{I}+O(\Delta_t)$. The only unitary operator $e^{i\delta}A$ that could possibly satisfy $(e^{i\delta}A)^\tau=\mathbb{I}$ for $\tau=1$ is the identity operator itself. But $e^{i\delta}A$ cannot even equal identity, as $A$ has $k_x$ and $k_y$ dependence from containing $\widehat{S_x}$ and $\widehat{S_y}$, and the angles are not permitted to depend on $k_x$ and $k_y$, so there is no possible way to cancel out the $k_x$ and $k_y$ dependence. Thus, there is no continuous time limit defined in Eq.~\eqref{eq:ham} for $\tau=1$.
\qed\end{proof}
\begin{lemma}\label{lma:constraints}
For the continuous time limit in Eq.~\eqref{eq:ham} to exist, $\theta_{0i}=2q\pi+\pi$ for any integer $q$ and $i=x$ or $y$, $\theta_{0j}=2\pi r$ for any integer $r$ and $j\neq i$, and $\delta=\frac{2\pi l}{\tau}-\frac{p\pi}{2}$ for odd integer $p$ and for any positive integer number $l$.
\end{lemma}
\begin{proof}
Following up on the constraint that $(e^{i\delta}A)^\tau=\mathbb{I}$ from Eq. \eqref{eq:reduced1sc^n}, let $U$ be the diagonalization matrix of $A$, and let $D$ be the matrix of eigenvalues of $A$. Then we have the following:
\begin{align}
    &(e^{i\delta}A)^\tau=e^{i\tau\delta}(U^{-1}DUU^{-1}DUU^{-1}DU\ldots)=e^{i\tau\delta}U^{-1}D^\tau U=\mathbb{I}\\
    &\rightarrow e^{i\tau\delta}D^\tau=
    UU^{-1}=\mathbb{I}\rightarrow e^{i\tau\delta}D^\tau=\mathbb{I}
\end{align}
 If we set the eigenvalues of $e^{i\delta}A$ equal to a $\tau^{th}$ root of unity $e^{2\pi i l/\tau}$ where $l=0,1,2,..$ (which is equivalent to the constraint $(e^{i\delta}A)^\tau=\mathbb{I}$), we will recover the following constraint equation for $\theta_{0x}$ and $\theta_{0y}$. Solving for $D$ by finding the eigenvalues of $A$, we have the following:
 \begin{equation}
    \begin{split}
     D=&\frac{1}{2}e^{-i(\phi_{0x}+\phi_{0y}+\zeta'_{0x}-\zeta'_{0y})/2}\\
     &\times\Big[\big((1+e^{i(\phi_{0x}+\phi_{0y}+\phi'_{0x}+\phi'_{0y})})\cos{\frac{\theta_{0x}}{2}}\cos{\frac{\theta_{0y}}{2}}-(e^{i(\phi_{0y}+\zeta'_{0x})}+e^{i(\phi_{0x}+\zeta'_{0y})})\sin{\frac{\theta_{0x}}{2}}\sin{\frac{\theta_{0y}}{2}}\big)\mathbb{I}\\
     &+\Bigg(\Big(\big(\sin{\frac{\theta_{0x}}{2}}\sin{\frac{\theta_{0y}}{2}}(e^{i(\phi_{0y}+\zeta'_{0x})}+e^{i(\phi_{0x}+\zeta'_{0y})})-\cos{\frac{\theta_{0x}}{2}}\cos{\frac{\theta_{0y}}{2}}(1+e^{i(\phi_{0x}+\phi_{0y}+\phi'_{0x}+\phi'_{0y})})\big)^2\\
     &-4e^{i(\phi_{0x}+\phi_{0y}+\phi'_{0x}+\phi'_{0y})}\Big)\Bigg)^{1/2}\sigma_z\Big]
    \end{split}
 \end{equation}
We see that $D$ is purely diagonal and is of the form $D=\begin{pmatrix}x+y&0\\0&x-y\end{pmatrix}$ for complex numbers $x,y$. We also see a repetition of certain terms in $D$, and can greatly reduce the verbosity of the equation by writing it the following way:
\begin{equation}
     D=\frac{1}{2Y^{1/2}}\Big[\big(U_1W_1-U_2W_2\big)\mathbb{I}+\Bigg(\Big(\big(W_2U_2-W_1U_1)\big)^2-4Y\Big)\Bigg)^{1/2}\sigma_z\Big]
\end{equation}
where
\begin{equation}
    \begin{split}
        &U_1=1+e^{i(\phi_{0x}+\phi_{0y}+\phi'_{0x}+\phi'_{0y})}\\
        &U_2=e^{i(\phi_{0y}+\zeta'_{0x})}+e^{i(\phi_{0x}+\zeta'_{0y})}\\
        &W_1=\cos{\frac{\theta_{0x}}{2}}\cos{\frac{\theta_{0y}}{2}}\\
        &W_2=\sin{\frac{\theta_{0x}}{2}}\sin{\frac{\theta_{0y}}{2}}\\
        &Y=e^{i(\phi_{0x}+\phi_{0y}+\phi'_{0x}+\phi'_{0y})}
    \end{split}
\end{equation}
Taking either non-zero component of $D$, setting it equal to $e^{2\pi il/\tau-\delta}$ (where $l=0,1,2,..$), and solving for either $W_1$ or $W_2$ yields the following constraint equation:
\begin{equation}\label{eq:constraint_function_simple}
    f(k_x,k_y)=W_1\cos(g(k_x,k_y))-W_1\cos(h(k_x,k_y))-c=0\footnote{This constraint reduces to the constraint obtained for $\theta_{0}$ in Ref.~\cite{manighalam} when the 1D limit is taken i.e. $\zeta_{0y},\theta_{0y},\phi_{0y},\delta=0$ (see Eq.~(A7) of Ref.~\cite{manighalam})}
\end{equation}
where 
\begin{equation}
    \begin{split}
        c&=\cos(\frac{2\pi l}{n}-\delta),\\
        g(k_x,k_y)&=\frac{\phi_{0x}+\phi_{0y}+\zeta'_{0x}(k_x)+\zeta'_{0y}(k_y)}{2},\\
        h(k_x,k_y)&=\frac{\phi_{0y}-\phi_{0x}+\zeta'_{0x}(k_x)-\zeta'_{0y}(k_y)}{2}.
    \end{split}
\end{equation}

Notice that the constraint Eq. \eqref{eq:constraint_function_simple} has to hold for all $k_x$ and $k_y$ and additionally, all derivatives of $f(k_x,k_y)$ with respect to $k_x$ and $k_y$ must equal zero as well. Using $\frac{\partial g(k_x,k_y)}{\partial k_x}=-1$, $\frac{\partial h(k_x,k_y)}{\partial k_x}=-1$, $\frac{\partial g(k_x,k_y)}{\partial k_y}=-1$, and $\frac{\partial h(k_x,k_y)}{\partial k_y}=1$ we obtain the derivative of $f(k_x,k_y)$ with respect to $k_x$ and $k_y$:
\begin{equation}
    \begin{split}
        &\frac{\partial f(k_x,k_y)}{\partial k_x}=W_1\sin(g(k_x,k_y))-W_2\sin(h(k_x,k_y))=0\\
        &\frac{\partial f(k_x,k_y)}{\partial k_y}=W_1\sin(g(k_x,k_y))+W_2\sin(h(k_x,k_y))=0.
    \end{split}
\end{equation}
For both of these equations to be true, we must have the following:
\begin{equation}
    \begin{split}
        &W_1\sin(g(k_x,k_y))=0\\
        &W_2\sin(h(k_x,k_y))=0.
    \end{split}
\end{equation}
Due to $\phi_{0i}$ and $\zeta_{0i}$ being parameters which cannot depend on $k_{i}$, it follows that $\sin(g(k_x,k_y))$ cannot equal zero for all values of $k_x$ and $k_y$, so the following must be true:
\begin{equation}\label{eq:constraints}
    \begin{split}
        &W_1=0\rightarrow\cos(\frac{\theta_{0x}}{2})\cos(\frac{\theta_{0y}}{2})=0\rightarrow \theta_{0i}=2q\pi+\pi \text{ for any integer $q$, and $i=x$ or $y$}\\
        &W_2=0\rightarrow\sin(\frac{\theta_{0x}}{2})\sin(\frac{\theta_{0y}}{2})=0\rightarrow \theta_{0j}=2\pi r \text{ for any integer $r$, and $j\neq i$}
    \end{split}
\end{equation}
In other words, $|(\theta_{0x}-\theta_{0y})\text{ mod }2\pi|=\pi$. This corresponds to either $C_x$ purely diagonal and $C_y$ purely off-diagonal, or vice-versa. Further, because $a,~b=0$, it must be true from Eq.~\eqref{eq:constraint_function_simple} that $c=0\rightarrow \cos(\frac{2\pi l}{\tau}-\delta)=0\rightarrow \delta=\frac{2\pi l}{n}-\frac{p\pi}{2}$ for odd integer $p$ and any positive integer number $l$.
\qed\end{proof}
\begin{lemma}\label{lma:n}
For the limit defined in Eq. \eqref{eq:ham} to be finite, $\tau$ must be even.
\end{lemma}
\begin{proof}
Consider $\tau$ even. Substituting our $\theta$ constraints from lemma~\ref{lma:constraints} into $(e^{i\delta}A)^{\tau}$, where $\tau=2w$ for some integer $w$, we find that $(e^{i\delta}A)^{2w}=(-e^{2i\delta}\mathbb{I})^w=\mathbb{I}$, as $A^2=-\mathbb{I}$ and $(-e^{2i\delta})^w=\mathbb{I}$ for all $w$. This implies that even powers of $\tau$ will satisfy $(e^{i\delta}A)^\tau=\mathbb{I}$. As for odd $\tau$, we can write $\tau=2s+1$ for some integer $s$ to obtain the following:
\begin{equation}
    (e^{i\delta}A)^\tau=(e^{i\delta}A)^{2s+1}=e^{i\delta}A
\end{equation}
This cannot equate to identity, as we showed in lemma \ref{lma:n=1} that for $\tau=1$ no parametrization of $A$ can make $e^{i\delta}A=\mathbb{I}$. Thus, $\tau$ must be even to have a finite continuum limit as defined in Eq.\eqref{eq:ham}.
\qed\end{proof}
Because the constraints on $\tau$ and $\theta_0$ hold true for all $l$ from the last two lemmas, we will choose $l=0$ for the remainder of the proof without loss of generality.
\begin{lemma}\label{lma:CTLimitH}
Let $\theta_{0x}=2\pi m+\nu \pi$ and $\theta_{0y}=2\pi t+(1-\nu)\pi$, where $\nu$ parametrizes the constraints in Eq.~\eqref{eq:constraints}. The continuous time limit will exist if $H$ is the following:
\begin{equation}
    \begin{split}
        \pmb{\nu=0:}\\
        H =\frac{1}{4}\big[&\theta_{1x}\big(S_x^2R_z(\zeta_{0x})+S_y^{2}R_z(2\zeta_{0y}+2\phi_{0x}-2\phi_{0y})\big)\\
        +&\theta_{1y}\big(R_z(-2\phi_{0y})+S_x^2S_y^{2}R_z(2\zeta_{0x}+2\phi_{0x}+2\zeta_{0y})\big)\big]\sigma_y
    \end{split}
\end{equation}

\begin{equation}
    \begin{split}
        \pmb{\nu=1:}\\
     \hspace{0.9cm}   H =\frac{1}{4}\big[&\theta_{1x}\big(S_x^2R_z(\zeta_{0x})+S_y^{-2}R_z(-2\zeta_{0y}-2\phi_{0x}-2\phi_{0y})\big)\\
        +&\theta_{1y}\big(R_z(-2\phi_{0y})+S_x^2S_y^{-2}R_z(2\zeta_{0x}-2\phi_{0x}-2\zeta_{0y})\big)\big]\sigma_y  
    \end{split}
\end{equation}

\end{lemma}
\begin{proof}
We begin by using that $A^2=-1$, $A^{-1}=-A$, and $(e^{i\delta}A)^\tau=\mathbb{I}$ to reduce Eq.~\eqref{eq:reduced1sc^n}: 
\begin{equation}\label{eq:reduced1sc^n_2}
    \begin{split}
    \widehat{W}^\tau&=e^{i\delta \tau}(A-\frac{i \Delta_t}{2} B)^\tau\\
    &=(e^{i\delta}A)^\tau(\mathbb{I}-\frac{i \Delta_t}{2}A^{-1}\sum_{j=0}^{\tau-1}A^{-j}BA^j+O(\Delta_t^2))\\
    &=\mathbb{I}+\frac{i \Delta_t}{2}A\sum_{j=0}^{\tau-1}(-1)^jA^{j}BA^j+O(\Delta_t^2).
    \end{split}
\end{equation}
Now we evaluate the sum by splitting it up into even and odd terms:
\begin{equation}
    \begin{split}
        A\sum_{j=0}^{\tau-1}(-1)^j A^{j}BA^j&=A(\sum_{j=odds}^{\tau-1}(-1)^j A^{j}BA^j+\sum_{j=evens}^{\tau-2}(-1)^j A^{j}BA^j)\\
        &=A\frac{\tau}{2}(-ABA+B)=\frac{\tau}{2}\{A,B\}
    \end{split}
\end{equation}

Leaving a detailed proof to Appendix~\ref{app:commutatorExpansion}, we have the following for $\{A,B\}$:
\begin{equation}
    \begin{split}
        \{A,B\}=&-\theta_{1y}(R_z(-2\phi_{0y})+R_z(2\zeta'_{0x}+2\phi_{0x}(-1)^\nu+2\zeta'_{0y}(-1)^\nu))\sigma_y\\
        &-\theta_{1x}(R_z(2\zeta'_{0x})+R_z(2\zeta'_{0y}(-1)^\nu-2\phi_{0y}+2\phi_{0x}(-1)^\nu))\sigma_y
    \end{split}
\end{equation}
Now we have the following for Eq.~\eqref{eq:reduced1sc^n_2}:
\begin{equation}
    \begin{split}
        \widehat{W}^\tau&=\mathbb{I}+\frac{i \Delta_t}{2}A\sum_{j=0}^{\tau-1}(-1)^jA^{j}BA^j+O(\Delta_t^2)\\
        &=\mathbb{I}+\frac{i\tau\Delta_t}{4}\{A,B\}+O(\Delta_t^2)\\
        &=\mathbb{I}-\frac{i\tau\Delta_t}{4}(\theta_{1y}(R_z(-2\phi_{0y})+R_z(2\zeta'_{0x}+2\phi_{0x}(-1)^\nu+2\zeta'_{0y}(-1)^\nu))\\
        &+\theta_{1x}(R_z(2\zeta'_{0x})+R_z(2\zeta'_{0y}(-1)^\nu-2\phi_{0y}+2\phi_{0x}(-1)^\nu)))\sigma_y+O(\Delta_t^2)
    \end{split}
\end{equation}
Now we evaluate the limit in Eq.~\eqref{eq:ham}:
\begin{equation}
    \begin{split}
        \widehat{H}&=i\lim\limits_{\Delta_t \to 0}\frac{(\widehat{S_x}C_x\widehat{S_y}C_y)^\tau-\mathbb{I}}{\tau\Delta_t}\\
        &=\frac{1}{4}(\theta_{1y}(R_z(-2\phi_{0y})+R_z(2\zeta'_{0x}+2\phi_{0x}(-1)^\nu+2\zeta'_{0y}(-1)^\nu))\\
        &+\theta_{1x}(R_z(2\zeta'_{0x})+R_z(2\zeta'_{0y}(-1)^\nu-2\phi_{0y}+2\phi_{0x}(-1)^\nu)))\sigma_y
    \end{split}
\end{equation}
Converting to real space, we get the following
\begin{equation}\label{eq:CTLimit_H_alpha}
    \begin{split}
        H=\frac{1}{4}\big[&\theta_{1x}\big(S_x^2R_z(\zeta_{0x})+S_y^{2(-1)^\nu}R_z(2\zeta_{0y}(-1)^\nu+2\phi_{0x}(-1)^\nu-2\phi_{0y})\big)\\
        +&\theta_{1y}\big(R_z(-2\phi_{0y})+S_x^2S_y^{2(-1)^\nu}R_z(2\zeta_{0x}+2\phi_{0x}(-1)^\nu+2\zeta_{0y}(-1)^\nu)\big)\big]\sigma_y    
    \end{split}
\end{equation}
The equations in lemma~\ref{lma:CTLimitH} are the same as Eq.~\eqref{eq:CTLimit_H_alpha} but for particular choices of $\nu$.
\qed\end{proof}

Note Eq.~\eqref{eq:CTLimit_H_alpha} reduces to the $H$ found in Ref.~\cite{manighalam} when the 1D limit is taken. We conclude the discussion with the following theorem encompassing our results:~

\begin{theorem}\label{thm:ctdsUnitaries}
    Let $C_j(\delta_j,\zeta_j,\theta_j,\phi_j)$ be the $2\times2$ unitary matrix in Eq. \eqref{eq:rotations}, with the set of angles {$\zeta_j$, $\theta_j$, $\phi_j$} parametrizing $C_j$ depending on $\Delta_t$ as: $\zeta_j=\zeta_{0j}+\zeta_{1j}\Delta_t$, $\theta_j=\theta_{0j}+\theta_{1j}\Delta_t$, and $\phi_j=\phi_{0j}+\phi_{1j}\Delta_t$, with $\phi_{0j},\zeta_{0j},\theta_{0j},\phi_{1j},\zeta_{1j},\theta_{1j}\in\mathbb{R}$ constants. The continuous time limit as defined in Eq.~\eqref{eq:real_space_ham} will exist for such a class of coins  if and only if $\theta_{0x}=2\pi m+\nu\pi$ and $\theta_{0y}=2\pi t+(1-\nu)\pi$ (for $\nu=0$ or $1$), $\delta=\delta_x+\delta_y=-\frac{p\pi}{2}$ (for odd integer $p$), and $n$ is even. The Hamiltonian obtained in such a limit, for each choice of $\nu$, is the following:
    \begin{equation}\label{eq:alpha0}
        \begin{split}
            H =\frac{1}{4}\big[&\theta_{1x}\big(S_x^2R_z(\zeta_{0x})+S_y^{2(-1)^\nu}R_z((-1)^\nu2\zeta_{0y}+(-1)^\nu 2\phi_{0x}-2\phi_{0y})\big)\\
            +&\theta_{1y}\big(R_z(-2\phi_{0y})+S_x^2S_y^{2(-1)^\nu}R_z(2\zeta_{0x}+2(-1)^\nu\phi_{0x}+2(-1)^\nu\zeta_{0y})\big)\big]\sigma_y
        \end{split}
    \end{equation}
   \end{theorem}

Notice that the above Hamiltonian is very general and encompasses the standard Dirac Hamiltonian on the 2D lattice where the cross-terms finite derivatives are not included.

\section{Plastic Quantum Walk in 2D+1}\label{sec:plasticQW}
In this section we will be analyzing the space of coins for which a continuous spacetime and continuous time exists. We begin by explicitly stating the problem. We use the same model as from section \ref{sec:model}, but now we consider an arbitrary $\Delta$. In Fourier space they are represented by the following operators:
\begin{equation}
    \begin{split}
        S_x&=e^{ik_x\Delta \sigma_z}\\
        S_y&=e^{ik_y\Delta \sigma_z}.
    \end{split}
\end{equation}
Now we parametrize the time and space steps the same way as in Ref.~\cite{molfetta2019quantum}:
\begin{equation}\label{eq:spacetimeParametrization}
    \Delta_t=\varepsilon\qquad\Delta=\varepsilon^a
\end{equation}
where $a\in[0,1]$, and $a=1$ is the continuous spacetime limit and $a=0$ is the continuous time limit. A word on the stroboscopic step ($\tau$). We wish to find a quantum walk which admits both a continuous spacetime and continuous time limit, and we know the $\tau$ must be even for the walk to admit a continuous time limit, so we only consider $\tau=2$ in this section. 
In the following, we will find for which parameters $\delta_j$, $\zeta_j$, $\theta_j$, and $\phi_j$ the continuous spacetime limit exists and converges to:
\begin{equation}\label{eq:ham2}
    \begin{split}
        \widehat{H}\hat{\Psi}(t)&=i\partial_t\hat{\Psi}(t)=i\lim\limits_{\varepsilon \to 0}\frac{\widehat{W}^2-\mathbb{I}}{2\varepsilon}\hat{\Psi}(t)\\
        &=i\lim\limits_{\varepsilon \to 0}\frac{(e^{ik_x\varepsilon^a\sigma_z}C_xe^{ik_y\varepsilon^a\sigma_z}C_y)^2-\mathbb{I}}{2\varepsilon}\hat{\Psi}(t).
    \end{split}
\end{equation}
We will then intersect these constraints with those found in the previous section to determine the quantum walks which admit a continuous time and continuous spacetime limit. In order to find the continuum limit in equation \eqref{eq:ham2}, let us first parametrize the $\theta$ parameter in the following way:
\begin{equation}\label{eq:scaling2}
    \begin{split}
        \theta_j&=\theta_{0j}+\theta_{1j}\varepsilon^b\\
    \end{split}
\end{equation}
where $b\in(0,1]$. Altogether, these jets define a family of QWs indexed by $\varepsilon$, whose embedding in spacetime, and defining angles, depend on $\varepsilon$. Notice that we only expand $\theta$ in powers of $\varepsilon^b$, this is because the Eq.\eqref{eq:alpha0} doesn't depend on the first order of $\zeta$ and $\phi$ and for plasticity we are interested in the smallest subset of constraint conditions to derive the continuum limit.\\
Now we expand $S_mC_m$ in powers of $\varepsilon$, where $m=x,y$:
\begin{equation}
    \begin{split}
        S_mC_m&=e^{\delta_m}R_z(\zeta'_{m})R_y(\theta_{m})R_z(\phi_{m})\\
        &=e^{\delta_m}R_z(\zeta'_{m})\sum_{n_m=-\infty}^{\infty}\frac{(-\frac{i\theta_{1m}\sigma_y}{2})^{n_m}}{n_m!}\varepsilon^{n_mb}R_y(\theta_{0m})R_z(\phi_{m})\\
        &=e^{\delta_m}R_z(\zeta_{m})\sum_{l_m=-\infty}^{\infty}\frac{(ik_m\sigma_z)^{l_m}}{l_m!}\varepsilon^{l_ma}\sum_{n_m=-\infty}^{\infty}\frac{(-\frac{i\theta_{1m}\sigma_y}{2})^{n_m}}{n_m!}\varepsilon^{n_mb}R_y(\theta_{0m})R_z(\phi_{m})\\
        &=e^{\delta_m}\sum_{l_m,n_m}\varepsilon^{l_ma+n_mb}\frac{(ik_m)^{l_m}(-\frac{i\theta_{1m}}{2})^{n_m}}{l_m!n_m!}R_z(\zeta_{m})\sigma_z^{l_m}\sigma_y^{n_m}R_y(\theta_{0m})R_z(\phi_{m})
    \end{split}
\end{equation}
Next we use the above equation to expand $S_xC_xS_yC_y$ in powers of $\varepsilon$:
\begin{equation}
    \begin{split}
        S_xC_xS_yC_y&=e^{i(\delta_x+\delta_y)}\sum_{\substack{l_x,n_x\\l_y,n_y}}\varepsilon^{a(l_x+l_y)+b(n_x+n_y)}\frac{(ik_x)^{l_x}(ik_y)^{l_y}(-\frac{i\theta_{1x}}{2})^{n_x}(-\frac{i\theta_{1y}}{2})^{n_y}}{l_x!l_y!n_x!n_y!}\\
        &\times R_z(\zeta_{x})\sigma_z^{l_x}\sigma_y^{n_x}R_y(\theta_{0x})R_z(\phi_{x})R_z(\zeta_{y})\sigma_z^{l_y}\sigma_y^{n_y}R_y(\theta_{0y})R_z(\phi_{y})\\
        &=e^{i(\delta_x+\delta_y)}\sum_{\substack{l_x,n_x\\l_y,n_y}}\varepsilon^{a(l_x+l_y)+b(n_x+n_y)}\nu_{l_xl_yn_xn_y}\hat{\Gamma}_{l_xl_yn_xn_y},
    \end{split}
\end{equation}
where
\begin{equation}
    \nu_{l_xl_yn_xn_y}=\frac{(ik_x)^{l_x}(ik_y)^{l_y}(-\frac{i\theta_{1x}}{2})^{n_x}(-\frac{i\theta_{1y}}{2})^{n_y}}{l_x!l_y!n_x!n_y!}
\end{equation}
and
\begin{equation}\label{eq:unitaryCoeffs}
    \begin{split}
        \hat{\Gamma}_{l_xl_yn_xn_y}&=R_z(\zeta_{x})\sigma_z^{l_x}\sigma_y^{n_x}R_y(\theta_{0x})R_z(\phi_{x})R_z(\zeta_{y})\sigma_z^{l_y}\sigma_y^{n_y}R_y(\theta_{0y})R_z(\phi_{y})
    \end{split}
\end{equation}
Now we have the following for $(S_xC_xS_yC_y)^2$:
\begin{equation}\label{eq:(sxcxsycy)^2}
    \begin{split}
        (S_xC_xS_yC_y)^2&=e^{2i(\delta_x+\delta_y)}(\sum_{\substack{l_x,n_x\\l_y,n_y}}\varepsilon^{a(l_x+l_y)+b(n_x+n_y)}\nu_{l_xl_yn_xn_y}\hat{\Gamma}_{l_xl_yn_xn_y})^2\\
        &=e^{2i(\delta_x+\delta_y)}\sum_{\substack{l_{1x},l_{1y}\\l_{2x},l_{2y}}}\sum_{\substack{n_{1x},n_{1y}\\n_{2x},n_{2y}}}\varepsilon^{a(l_{1x}+l_{1y}+l_{2x}+l_{2y})+b(n_{1x}+n_{1y}+n_{2x}+n_{2y})}\\
        &\times \nu_{l_{1x}l_{1y}n_{1x}n_{1y}}\nu_{l_{2x}l_{2y}n_{2x}n_{2y}}\hat{\Gamma}_{l_{1x}l_{1y}n_{1x}n_{1y}}\hat{\Gamma}_{l_{2x}l_{2y}n_{2x}n_{2y}}\\
        &=e^{2i(\delta_x+\delta_y)}\nu_{0000}^2\hat{\Gamma}_{0000}^2\\
        &+e^{2i(\delta_x+\delta_y)}\sum_{\substack{l_{1x},l_{1y}\\l_{2x},l_{2y}\\\neq (0,0,0,0)}}\sum_{\substack{n_{1x},n_{1y}\\n_{2x},n_{2y}\\\neq (0,0,0,0)}}\varepsilon^{a(l_{1x}+l_{1y}+l_{2x}+l_{2y})+b(n_{1x}+n_{1y}+n_{2x}+n_{2y})}\\
        &\times \nu_{l_{1x}l_{1y}n_{1x}n_{1y}}\nu_{l_{2x}l_{2y}n_{2x}n_{2y}}\hat{\Gamma}_{l_{1x}l_{1y}n_{1x}n_{1y}}\hat{\Gamma}_{l_{2x}l_{2y}n_{2x}n_{2y}}.
    \end{split}
\end{equation}

Now we have the following lemmas:
\begin{lemma}
    A DTQW will allow both a continuous spacetime limit (as in Eq.~\eqref{eq:ham2}) and a continuous time limit (as in Eq.~\eqref{eq:real_space_ham}) if and only if $\theta_{0x}=2\pi m$ and $\theta_{0y}=2\pi t+\pi$.
\end{lemma}
\begin{proof}
    For the limit in Eq.~\eqref{eq:ham2} to exist, $e^{2i(\delta_x+\delta_y)}\nu_{0000}^2\hat{\Gamma}_{0000}^2$ must equal identity in Eq.~\eqref{eq:(sxcxsycy)^2}, so we have the following constraint:
    \begin{equation}\label{eq:ctcs_angle_constraints}
        \begin{split}
            &e^{2i(\delta_x+\delta_y)}\nu_{0000}^2\hat{\Gamma}_{0000}^2\\
            &=e^{2i(\delta_x+\delta_y)}R_z(\zeta_{x})R_y(\theta_{0x})R_z(\phi_{x})R_z(\zeta_{y})R_y(\theta_{0y})R_z(\phi_{y})\\
            &\times R_z(\zeta_{x})R_y(\theta_{0x})R_z(\phi_{x})R_z(\zeta_{y})R_y(\theta_{0y})R_z(\phi_{y})=\mathbb{I}
        \end{split}
    \end{equation}
    Going though the same process of finding the eigenvalues of $\nu_{0000}^2\hat{\Gamma}_{0000}$ and setting them equal to a root of unity, we find that constraint yields a similar equation as from the continuous time limit in lemma~\ref{lma:constraints}.It is the following (where $\delta=\delta_x+\delta_y)$):
    \begin{equation}\label{eq:constraint_function_simple2}
        \begin{split}
            f&=\cos(\frac{\theta_{0x}}{2})\cos(\frac{\theta_{0y}}{2})\cos(\frac{\phi_{x}+\phi_{y}+\zeta_{x}+\zeta_{y}}{2})\\
            &-\sin(\frac{\theta_{0x}}{2})\sin(\frac{\theta_{0y}}{2})\cos(\frac{\phi_{y}-\phi_{x}+\zeta_{x}-\zeta_{y}}{2})\\
            &-\cos(\frac{2\pi l}{n}-\delta)=0.        
        \end{split}
    \end{equation}
    As from section~\ref{sec:limits}, we choose $\delta$ such that $\cos(\frac{2\pi l}{n}-\delta)=0$. We see a new set of constraints are available than those found in section~\ref{sec:limits}. The new types of constraints involve $\cos(\frac{\phi_{x}+\phi_{y}+\zeta_{x}+\zeta_{y}}{2})$ equalling zero and any other of the three products in $\sin(\frac{\theta_{0x}}{2})\sin(\frac{\theta_{0y}}{2})\cos(\frac{\phi_{y}-\phi_{x}+\zeta_{x}-\zeta_{y}}{2})$ equalling zero, or $\cos(\frac{\phi_{y}-\phi_{x}+\zeta_{x}-\zeta_{y}}{2})$ equalling zero and any other of the three products in $\cos(\frac{\theta_{0x}}{2})\cos(\frac{\theta_{0y}}{2})\cos(\frac{\phi_{x}+\phi_{y}+\zeta_{x}+\zeta_{y}}{2})$ equalling zero. Another possible set of parameter constraints from Eq.~\eqref{eq:constraint_function_simple2} is $\theta_{0x}=2\pi m$ and $\theta_{0y}=2\pi t+\pi$, which are included in the set of constraints found in the continuous time limit of section~\ref{sec:limits}. For the purpose of this work, we only consider this last set of parameter constraints, as our goal is to develop a DTQW which admits both a continuous time limit as well as a continuous spacetime limit.
\qed\end{proof}
\begin{lemma}
    A DTQW will allow both a continuous spacetime limit (as in Eq.~\eqref{eq:ham2}) and a continuous time limit (as in Eq.~\eqref{eq:real_space_ham}) if and only if $a(l_{1x}+l_{1y}+l_{2x}+l_{2y})+b(n_{1x}+n_{1y}+n_{2x}+n_{2y})=1$ and $a$, $b\in\mathbb{Q}$.
\end{lemma}
\begin{proof}
    We will determine how choices of $a$ and $b$ change Eq.~\eqref{eq:(sxcxsycy)^2}. The only terms in Eq.~\eqref{eq:(sxcxsycy)^2} that will contribute to the continuum limit will be those of order $\varepsilon$, which yields a constraint concerning which terms in the sum will be non-zero after the continuum limit is taken, given a choice of $a$ and $b$:
    \begin{equation}\label{eq:ab_constraint}
        a(l_{1x}+l_{1y}+l_{2x}+l_{2y})+b(n_{1x}+n_{1y}+n_{2x}+n_{2y})=1.
    \end{equation}
    Also, since $l_{vm},n_{vm}\in\mathbb{Z}_+$ (where $v=1,2$ and $m=x,y$), $a$ and $b$ must be in $\mathbb{Q}$ for this equation to hold.
\qed\end{proof}
\begin{lemma}
    A DTQW will allow both a continuous spacetime limit (as in Eq.~\eqref{eq:ham2}) and a continuous time limit (as in Eq.~\eqref{eq:real_space_ham}) if and only if 
    \begin{equation}\label{eq:divergenceConstraint}
    \begin{split}
        \sum_{\substack{l_{1x},l_{1y}\\l_{2x},l_{2y}\\\neq (0,0,0,0)}}\sum_{\substack{n_{1x},n_{1y}\\n_{2x},n_{2y}\\\neq (0,0,0,0)}}&\varepsilon^{a(l_{1x}+l_{1y}+l_{2x}+l_{2y})+b(n_{1x}+n_{1y}+n_{2x}+n_{2y})}\\
        \times& \nu_{l_{1x}l_{1y}n_{1x}n_{1y}}\nu_{l_{2x}l_{2y}n_{2x}n_{2y}}\hat{\Gamma}_{l_{1x}l_{1y}n_{1x}n_{1y}}\hat{\Gamma}_{l_{2x}l_{2y}n_{2x}n_{2y}}\\
        \times& \mathcal{H}(1-a(l_{1x}+l_{1y}+l_{2x}+l_{2y})-b(n_{1x}+n_{1y}+n_{2x}+n_{2y}))=0
    \end{split}
    \end{equation}
    where \[ \mathcal{H}(x)=\begin{cases} 
          0 & x\leq 0 \\
          1 & x>0. 
       \end{cases}
    \]
\end{lemma}
\begin{proof}
     Since terms of order $\varepsilon^f$, where $0<f<1$, diverge when $\varepsilon\to 0$ in Eq.~\eqref{eq:ham2}, terms of order $\varepsilon^f$ to sum to zero for the limit to exist. This is summarized in the following constraint:
    \begin{equation}
        \begin{split}
            \sum_{\substack{l_{1x},l_{1y}\\l_{2x},l_{2y}\\\neq (0,0,0,0)}}\sum_{\substack{n_{1x},n_{1y}\\n_{2x},n_{2y}\\\neq (0,0,0,0)}}&\varepsilon^{a(l_{1x}+l_{1y}+l_{2x}+l_{2y})+b(n_{1x}+n_{1y}+n_{2x}+n_{2y})}\\
            \times& \nu_{l_{1x}l_{1y}n_{1x}n_{1y}}\nu_{l_{2x}l_{2y}n_{2x}n_{2y}}\hat{\Gamma}_{l_{1x}l_{1y}n_{1x}n_{1y}}\hat{\Gamma}_{l_{2x}l_{2y}n_{2x}n_{2y}}\\
            \times& \mathcal{H}(1-a(l_{1x}+l_{1y}+l_{2x}+l_{2y})-b(n_{1x}+n_{1y}+n_{2x}+n_{2y}))=0
        \end{split}
    \end{equation}
    where \[ \mathcal{H}(x)=\begin{cases} 
          0 & x\leq 0 \\
          1 & x>0. 
       \end{cases}
    \]
\qed\end{proof}
Upon further analysis of Eq.~\eqref{eq:divergenceConstraint}, we see that $l_{vm}$ produces a spatial derivative with respect to $m$ in the term, so cross terms with multiple derivatives will be terms in the sum with multiple non-zero $l_{vm}$'s. $n_{vm}$ determines the presence of the driving parameter $\theta_{1m}$ in the term. 
Now for our last lemma:
\begin{lemma}
    A DTQW will allow both a continuous spacetime limit (as in Eq.~\eqref{eq:ham2}) and a continuous time limit (as in Eq.~\eqref{eq:real_space_ham}) if and only if $H$ is the following:
    \begin{equation}\label{eq:csct_ham_final}
    \begin{split}
        H\Psi(x,y,t)&=\frac{1}{2}\sum_{\substack{l_{1x},l_{1y}\\l_{2x},l_{2y}\\\neq (0,0,0,0)}}\sum_{\substack{n_{1x},n_{1y}\\n_{2x},n_{2y}\\\neq (0,0,0,0)}}\delta_{a(l_{1x}+l_{1y}+l_{2x}+l_{2y})+b(n_{1x}+n_{1y}+n_{2x}+n_{2y})~1}\\
        &\times \nu'_{l_{1x}l_{1y}n_{1x}n_{1y}}\nu'_{l_{2x}l_{2y}n_{2x}n_{2y}}\hat{\Gamma}_{l_{1x}l_{1y}n_{1x}n_{1y}}\hat{\Gamma}_{l_{2x}l_{2y}n_{2x}n_{2y}}\Psi(x,y,t)
    \end{split}
    \end{equation}
    where
    \begin{equation}
        \nu'_{l_xl_yn_xn_y}=\frac{\partial_x^{l_x}\partial_y^{l_y}(-\frac{i\theta_{1x}}{2})^{n_x}(-\frac{i\theta_{1y}}{2})^{n_y}}{l_x!l_y!n_x!n_y!}
    \end{equation}
\end{lemma}
\begin{proof}
    We use Eq.~\eqref{eq:ham2} to evaluate the limit. Using that $\mathcal{F}^{-1}(ik_m)=\partial_m$, we have the following (with $\delta_{ij}$ being the Kronecker delta):
    \begin{equation}
        \begin{split}
            H\Psi(x,y,t)&=\frac{1}{2}\sum_{\substack{l_{1x},l_{1y}\\l_{2x},l_{2y}\\\neq (0,0,0,0)}}\sum_{\substack{n_{1x},n_{1y}\\n_{2x},n_{2y}\\\neq (0,0,0,0)}}\delta_{a(l_{1x}+l_{1y}+l_{2x}+l_{2y})+b(n_{1x}+n_{1y}+n_{2x}+n_{2y})~1}\\
            &\times \nu'_{l_{1x}l_{1y}n_{1x}n_{1y}}\nu'_{l_{2x}l_{2y}n_{2x}n_{2y}}\hat{\Gamma}_{l_{1x}l_{1y}n_{1x}n_{1y}}\hat{\Gamma}_{l_{2x}l_{2y}n_{2x}n_{2y}}\Psi(x,y,t)
        \end{split}
    \end{equation}
    where
    \begin{equation}
        \nu'_{l_xl_yn_xn_y}=\frac{\partial_x^{l_x}\partial_y^{l_y}(-\frac{i\theta_{1x}}{2})^{n_x}(-\frac{i\theta_{1y}}{2})^{n_y}}{l_x!l_y!n_x!n_y!}
    \end{equation}
\end{proof}

As in the previous section, we conclude the discussion with the following theorem encompassing our result:
\begin{theorem}\label{thm:2dContSpaceTimeLimit}
    Let the lattice spacing and time steps of the 2D DTQW be parametrized by infinitesimal parameter $\varepsilon$ as in Eq.~\eqref{eq:spacetimeParametrization}. Let $\tau=2$ and let $C_j(\delta_j,\zeta_j,\theta_j,\phi_j)$ be the $2\times2$ unitary matrix in Eq. \eqref{eq:rotations}, with only the angle $\theta_j$ depending on $\varepsilon$ in the following way: $\theta_j=\theta_{0j}+\theta_{1j}\varepsilon^b$ with $\theta_{0j},\theta_{1j}\in\mathbb{R}$, $b\in(0,1]$. A DTQW will allow both a continuous spacetime limit (as defined in Eq.~\eqref{eq:ham2})  and continuous time limit (as defined in Eq.~\eqref{eq:real_space_ham} with $\tau=2$) if and only if the following 4 constraints are met:
    \begin{equation}\label{eq:thm2_Constraints}
        \begin{split}
            &1.~\theta_{0x}=2\pi m \text{ and } \theta_{0y}=2\pi t+\pi\\
            &2.~a(l_{1x}+l_{1y}+l_{2x}+l_{2y})+b(n_{1x}+n_{1y}+n_{2x}+n_{2y})=1.\\
            &3.~a,~b\in\mathbb{Q}\\
            &4.~\sum_{\substack{l_{1x},l_{1y}\\l_{2x},l_{2y}\\\neq(0,0,0,0)}}\sum_{\substack{n_{1x},n_{1y}\\n_{2x},n_{2y}\\\neq(0,0,0,0)}}\varepsilon^{a(l_{1x}+l_{1y}+l_{2x}+l_{2y})+b(n_{1x}+n_{1y}+n_{2x}+n_{2y})}\\
            &~~~\times \nu_{l_{1x}l_{1y}n_{1x}n_{1y}}\nu_{l_{2x}l_{2y}n_{2x}n_{2y}}\hat{\Gamma}_{l_{1x}l_{1y}n_{1x}n_{1y}}\hat{\Gamma}_{l_{2x}l_{2y}n_{2x}n_{2y}}\\
            &~~~\times \mathcal{H}(1-a(l_{1x}+l_{1y}+l_{2x}+l_{2y})-b(n_{1x}+n_{1y}+n_{2x}+n_{2y}))=0
        \end{split}
    \end{equation}
    The Hamiltonian obtained in such a limit is the following:
    \begin{equation}\label{eq:thm2_Ham}
        \begin{split}
            H\Psi(x,y,t)&=\frac{1}{2}\sum_{\substack{l_{1x},l_{1y}\\l_{2x},l_{2y}\\\neq (0,0,0,0)}}\sum_{\substack{n_{1x},n_{1y}\\n_{2x},n_{2y}\\\neq (0,0,0,0)}}\delta_{a(l_{1x}+l_{1y}+l_{2x}+l_{2y})+b(n_{1x}+n_{1y}+n_{2x}+n_{2y})~1}\\
            &\times \nu'_{l_{1x}l_{1y}n_{1x}n_{1y}}\nu'_{l_{2x}l_{2y}n_{2x}n_{2y}}\hat{\Gamma}_{l_{1x}l_{1y}n_{1x}n_{1y}}\hat{\Gamma}_{l_{2x}l_{2y}n_{2x}n_{2y}}\Psi(x,y,t)
        \end{split}
    \end{equation}
    where
    \begin{equation}
        \nu'_{l_xl_yn_xn_y}=\frac{\partial_x^{l_x}\partial_y^{l_y}(-\frac{i\theta_{1x}}{2})^{n_x}(-\frac{i\theta_{1y}}{2})^{n_y}}{l_x!l_y!n_x!n_y!}
    \end{equation}
    and $\hat{\Gamma}$ is as defined in Eq.~\eqref{eq:unitaryCoeffs}
\end{theorem}
The form of the limit in Eq.~\eqref{eq:csct_ham_final} is very powerful, as it identifies the type of PDE obtained for any possible choice of $\theta$ and $\Delta_x=\Delta_y$ scaling dependence of $\varepsilon$.

\subsection{Example with $a=b=1/2$}
As an example, we analyze the $a=b=\frac{1}{2}$ scenario. For this case, the only terms which are not zero by the Kronecker delta in Eq.~\eqref{eq:csct_ham_final} are those with two of the $l_{vm}$s equalling 1 and the $n_{vm}$s equalling 0 (6 terms), one $l_{vm}$ equalling 1 and one $n_{vm}$ equalling 1 (16 terms), two $n_{vm}$s equalling 1 and the $l_{vm}$s equalling 0 (6 terms), and one $l_{vm}$ or $n_{vm}$ equalling 2 with the rest equalling zero (8 terms). In the following section, we will be analyzing each of these terms and applying constraints to them to uncover the PDE in this continuous spacetime limit. The main constraints we will be focusing on are the $\theta_{0x}=2\pi m$ and $\theta_{0y}=2\pi t+\pi$ constraints and the constraints from Eq.~\eqref{eq:divergenceConstraint}.

We begin by writing the non-divergence constraint from Eq.~\eqref{eq:divergenceConstraint}, adapted to our $a=b=1/2$ example:
\begin{equation}
    \begin{split}
        \varepsilon^{\frac{1}{2}}(&G_{10000000}+G_{0100000}+G_{00100000}+G_{00010000}\\
        +&G_{00001000}+G_{00000100}+G_{00000010}+G_{00000001})=0.
    \end{split}
\end{equation}
where
\begin{equation}
    G_{l_{1x}l_{2x}l_{1y}l_{2y}n_{1x}n_{2x}n_{1y}n_{2y}}=\nu_{l_{1x}l_{1y}n_{1x}n_{1y}}\nu_{l_{2x}l_{2y}n_{2x}n_{2y}}\hat{\Gamma}_{l_{1x}l_{1y}n_{1x}n_{1y}}\hat{\Gamma}_{l_{2x}l_{2y}n_{2x}n_{2y}}
\end{equation}
Collecting terms of order $\theta_{1x}$, $\theta_{1y}$, $\partial_x$, and $\partial_y$ and setting each to zero, we obtain the following constraint equations (where $a_1=\phi_{x}+\zeta_{y}$ and $a_2=\phi_{y}+\zeta_{x}$):
\begin{equation}\label{eq:div_constraints_list}
    \begin{split}
        &\text{$\theta_{1x}$: } R_z(a_1)R_y(\theta_{0y})R_z(a_2)+R_z(-a_1)R_y(\theta_{0y})R_z(-a_2)=0\\
        &\text{$\theta_{1y}$: } R_z(a_2)R_y(\theta_{0x})R_z(a_1)+R_z(-a_2)R_y(\theta_{0x})R_z(-a_1)=0\\
        &\text{$\partial_x$: } R_y(\theta_{0x})R_z(a_1)R_y(\theta_{0y})+R_y(-\theta_{0x})R_z(a_1)R_y(-\theta_{0y})=0\\
        &\text{$\partial_y$: } R_y(\theta_{0y})R_z(a_2)R_y(\theta_{0x})+R_y(-\theta_{0y})R_z(a_2)R_y(-\theta_{0x})=0.
    \end{split}
\end{equation}
When plugging in $\theta_{0x}=2\pi m$ and $\theta_{0y}=2\pi t+\pi$ in these equations, we see that first two equations in Eqs.~\eqref{eq:div_constraints_list} are satisfied only if $a_1=\frac{\pi}{2}(\alpha+1)$ and $a_2=\frac{\pi}{2}\beta$ (for integer $\alpha$ and $\beta$). We will be referring to these constraints when analyzing the terms with two of the $n_{vm}$s equalling one and terms with one $n_{vm}$ equalling 2.

Now we analyze the terms themselves. We first analyze the terms with two of the $l_{vm}$s equalling one and terms with one $l_{vm}$ equalling 2. We will see that the one term with $l_{1m}=1$ and $l_{2m}=1$ will cancel with the term with $l_{1m}=2$ and the term with $l_{2m}=2$ when the constraints $\theta_x=2\pi m$, $\theta_y=2\pi t+\pi$ are used. We begin by writing $\nu'\nu'\hat{\Gamma}\hat{\Gamma}$ for terms with $n_{vm}=0$ (we call it $t_{l_{1x}l_{2x}l_{1y}l_{2y}}$):
\begin{equation}\label{eq:t_no_n}
    \begin{split}
        &t_{l_{1x}l_{2x}l_{1y}l_{2y}}=\nu'_{l_{1x}l_{1y}00}\nu'_{l_{2x}l_{2y}00}\hat{\Gamma}_{l_{1x}l_{1y}00}\hat{\Gamma}_{l_{2x}l_{2y}00}\\
        &=\frac{\partial_x^{l_{1x}+l_{2x}}\partial_y^{l_{1y}+l_{2y}}}{l_{1x}!l_{2x}!l_{1y}!l_{2y}!}\\
        &\times R_z(\zeta_{x})\sigma_z^{l_{1x}}R_y(\theta_{0x})R_z(\phi_{x})R_z(\zeta_{y})\sigma_z^{l_{1y}}R_y(\theta_{0y})R_z(\phi_{y})\\
        &\times R_z(\zeta_{x})\sigma_z^{l_{2x}}R_y(\theta_{0x})R_z(\phi_{x})R_z(\zeta_{y})\sigma_z^{l_{2y}}R_y(\theta_{0y})R_z(\phi_{y}).
    \end{split}
\end{equation}
We will analyze terms proportional to $\partial_x^2$, $\partial_y^2$, and $\partial_x\partial_y$. For the $\partial_x^2$ terms, the relevant part of Eq.~\eqref{eq:t_no_n} is the following:
\begin{equation}
    \begin{split}
        t_{l_{1x}l_{2x}00}\propto \sigma_z^{l_{1x}}R_y(\theta_{0x})R_z(\phi_{x})R_z(\zeta_{y})R_y(\theta_{0y})\sigma_z^{l_{2x}}/l_{1x}!l_{2x}!
    \end{split}
\end{equation}
Now we compute the relevant part of the sum $t_{1100}+t_{2000}+t_{0200}$ (that is, the sum of the terms proportional to $\partial_x^2$):
\begin{equation}
    \begin{split}
        t_{1100}+t_{2000}+t_{0200}&\propto \sigma_zR_y(\theta_{0x})R_z(\phi_{x})R_z(\zeta_{y})R_y(\theta_{0y})\sigma_z\\
        &+\sigma_z^2R_y(\theta_{0x})R_z(\phi_{x})R_z(\zeta_{y})R_y(\theta_{0y})/2\\
        &+R_y(\theta_{0x})R_z(\phi_{x})R_z(\zeta_{y})R_y(\theta_{0y})\sigma_z^2/2\\
        &=R_y(-\theta_{0x})R_z(\phi_{x})R_z(\zeta_{y})R_y(-\theta_{0y})\\
        &+R_y(\theta_{0x})R_z(\phi_{x})R_z(\zeta_{y})R_y(\theta_{0y})
    \end{split}
\end{equation}
Applying the constraints $\theta_x=2\pi m$, $\theta_y=2\pi t+\pi$, we obtain the following:
\begin{equation}
    \begin{split}
        &t_{1100}+t_{2000}+t_{0200}\\
        &\propto (-1)^mR_z(\phi_{x})R_z(\zeta_{y})(i\sigma_y)+(-1)^mR_z(\phi_{x})R_z(\zeta_{y})(-i\sigma_y)=0.
    \end{split}
\end{equation}
A cancellation also occurs for terms proportional to $\partial_y^2$ by the same reasoning. Concerning terms proportional to $\partial_x\partial_y$, they will not be present if any constraint with $\theta_x$ or $\theta_y$ being equal to an integer multiple of $\pi$ is used (see Appendix~\ref{app:cross_term_constraints}).

Now we analyze the terms with two of the $n_{vm}$s equalling one and terms with one $n_{vm}$s equalling 2. These terms are interpreted as mass terms in the continuum limit, as they are not proportional to any derivatives. We will see that a constraint in Eq.~\eqref{eq:divergenceConstraint} which enforces no divergences when $\varepsilon\to0$ cancels these terms, even before any constraints from Eq.~\eqref{eq:constraint_function_simple2} are used. As before, we begin by writing $\nu'\nu'\hat{\Gamma}\hat{\Gamma}$ for terms with $l_{vm}=0$ (we call it $\tilde{t}_{n_{1x}n_{2x}n_{1y}n_{2y}}$):
\begin{equation}\label{eq:n_no_t}
    \begin{split}
        &\tilde{t}_{n_{1x}n_{2x}n_{1y}n_{2y}}=\nu'_{00n_{1x}n_{1y}}\nu'_{00n_{2x}n_{2y}}\hat{\Gamma}_{00n_{1x}n_{1y}}\hat{\Gamma}_{00n_{2x}n_{2y}}\\
        &=\frac{(-\frac{i\theta_{1x}}{2})^{n_{1x}+n_{2x}}(-\frac{i\theta_{1y}}{2})^{n_{1y}+n_{2y}}}{n_{1x}!n_{2x}!n_{1y}!n_{2y}!}\\
        &\times R_z(\zeta_{x})\sigma_y^{n_{1x}}R_y(\theta_{0x})R_z(\phi_{x})R_z(\zeta_{y})\sigma_y^{n_{1y}}R_y(\theta_{0y})R_z(\phi_{y})\\
        &\times R_z(\zeta_{x})\sigma_y^{n_{2x}}R_y(\theta_{0x})R_z(\phi_{x})R_z(\zeta_{y})\sigma_y^{n_{2y}}R_y(\theta_{0y})R_z(\phi_{y}).
    \end{split}
\end{equation}
We will analyze terms proportional to $\theta_{1x}^2$, $\theta_{1y}^2$, and $\theta_{1x}\theta_{1y}$. For terms proportional to $\theta_{1x}^2$, the relevant part of Eq.~\eqref{eq:n_no_t} is the following:
\begin{equation}
    \tilde{t}_{n_{1x}n_{2x}00}\propto\sigma_y^{n_{1x}}R_z(\phi_{x})R_z(\zeta_{y})R_y(\theta_{0y})R_z(\phi_{y}) R_z(\zeta_{x})\sigma_y^{n_{2x}}/n_{1x}!n_{2x}!
\end{equation}
Now we compute the relevant part of the sum $\tilde{t}_{1100}+\tilde{t}_{2000}+\tilde{t}_{0200}$ (that is, the sum of the terms proportional to $\theta_{1x}^2$):
\begin{equation}
    \begin{split}
        \tilde{t}_{1100}+\tilde{t}_{2000}+\tilde{t}_{0200}&\propto\sigma_yR_z(\phi_{x})R_z(\zeta_{y})R_y(\theta_{0y})R_z(\phi_{y}) R_z(\zeta_{x})\sigma_y\\
        &+\sigma_y^2R_z(\phi_{x})R_z(\zeta_{y})R_y(\theta_{0y})R_z(\phi_{y}) R_z(\zeta_{x})/2\\
        &+R_z(\phi_{x})R_z(\zeta_{y})R_y(\theta_{0y})R_z(\phi_{y}) R_z(\zeta_{x})\sigma_y^2/2\\
        &=R_z(-\phi_{x})R_z(-\zeta_{y})R_y(\theta_{0y})R_z(-\phi_{y}) R_z(-\zeta_{x})\\
        &+R_z(\phi_{x})R_z(\zeta_{y})R_y(\theta_{0y})R_z(\phi_{y})R_z(\zeta_{x})
    \end{split}
\end{equation}
These terms do not cancel when the constraints $\theta_x=2\pi m$ and $\theta_y=2\pi t+\pi$ are used, rather they cancel when the $\theta_{1y}$ non-divergence constraint from Eq.~\eqref{eq:div_constraints_list} is imposed.

Next, we analyze the terms with one $l_{vm}$ and one $n_{nm}$. First we write $\nu'\nu'\hat{\Gamma}\hat{\Gamma}$ for all $l_{vm}$ and $n_{vm}$ (we denote it as $\hat{t}_{l_{1x}l_{2x}l_{1y}l_{2y}n_{1x}n_{2x}n_{1y}n_{2y}}$):
\begin{equation}
    \begin{split}
        \hat{t}_{l_{1x}l_{2x}l_{1y}l_{2y}n_{1x}n_{2x}n_{1y}n_{2y}}&=\partial_x^{l_{1x}+l_{2x}}\partial_y^{l_{1y}+l_{2y}}(\frac{-i\theta_{1x}}{2})^{n_{1x}+n_{2x}}(\frac{-i\theta_{1y}}{2})^{n_{1y}+n_{2y}}\\
        &\times R_z(\zeta_x)\sigma_z^{l_{1x}}\sigma_y^{n_{1x}}R_y(\theta_{0x})R_z(a_1)\sigma_z^{l_{1y}}\sigma_y^{n_{1y}}R_y(\theta_{0y})R_z(a_2)\\
        &\times \sigma_z^{l_{2x}}\sigma_y^{n_{2x}}R_y(\theta_{0x})R_z(a_1)\sigma_z^{l_{2y}}\sigma_y^{n_{2y}}R_y(\theta_{0y})R_z(\phi_y).
    \end{split}
\end{equation}
Now we reduce the above expression by plugging in the constraints $\theta_x=2\pi m$ and $\theta_y=2\pi t+\pi$:
\begin{equation}
    \begin{split}
        \hat{t}_{l_{1x}l_{2x}l_{1y}l_{2y}n_{1x}n_{2x}n_{1y}n_{2y}}&=-\partial_x^{l_{1x}+l_{2x}}\partial_y^{l_{1y}+l_{2y}}(\frac{-i\theta_{1x}}{2})^{n_{1x}+n_{2x}}(\frac{-i\theta_{1y}}{2})^{n_{1y}+n_{2y}}\\
        &\times R_z(\zeta_x)\sigma_z^{l_{1x}}\sigma_y^{n_{1x}}R_z(a_1)\sigma_z^{l_{1y}}\sigma_y^{n_{1y}+1}R_z(a_2)\sigma_z^{l_{2x}}\sigma_y^{n_{2x}}\\
        &\times R_z(a_1)\sigma_z^{l_{2y}}\sigma_y^{n_{2y}+1}R_z(\phi_y).
    \end{split}
\end{equation}
Thus we collect terms proportional to $\theta_{1x}\partial_x$, $\theta_{1y}\partial_y$, $\theta_{1x}\partial_y$, and $\theta_{1y}\partial_x$. First $\theta_{1x}\partial_x$:
\begin{equation}
    \begin{split}
        \sum_{\substack{l_{1x},l_{2x}\\n_{1x},n_{2x}\\=\{0,1\}}} \hat{t}_{l_{1x}l_{2x}00n_{1x}n_{2x}00}&=\partial_x\frac{i\theta_{1x}}{2}R_z(\zeta_x)\sigma_z^{l_{1x}}\sigma_y^{n_{1x}}R_z(a_1)\sigma_yR_z(a_2)\sigma_z^{l_{2x}}\sigma_y^{n_{2x}}\\
        &\times R_z(a_1)\sigma_yR_z(\phi_y)\\
        &=i\theta_{1x}\partial_x\sigma_zR_z(2(\phi_y+\zeta_x)).
    \end{split}
\end{equation}
Finally the $\theta_{1y}\partial_y$ term:
\begin{equation}
    \begin{split}
        \sum_{\substack{l_{1y},l_{2y}\\n_{1y},n_{2y}\\=\{0,1\}}} \hat{t}_{00l_{1y}l_{2y}00n_{1y}n_{2y}}&=\partial_y\frac{i\theta_{1y}}{2}R_z(\zeta_x)R_z(a_1)\sigma_z^{l_{1y}}\sigma_y^{n_{1y}+1}R_z(a_2)\\
        &\times R_z(a_1)\sigma_z^{l_{2y}}\sigma_y^{n_{2y}+1}R_z(\phi_y)\\
        &=i\theta_{1y}\partial_y\sigma_z\sigma_yR_z(2\phi_y).
    \end{split}
\end{equation}
The $\theta_{1x}\partial_y$ term:
\begin{equation}
    \begin{split}
        \sum_{\substack{l_{1x},l_{2x}\\n_{1y},n_{2y}\\=\{0,1\}}} \hat{t}_{l_{1x}l_{2x}0000n_{1y}n_{2y}}&=\partial_y\frac{i\theta_{1x}}{2}R_z(\zeta_x)\sigma_z^{l_{1x}}R_z(a_1)\sigma_y^{n_{1y}+1}R_z(a_2)\sigma_z^{l_{2x}}\\
        &\times R_z(a_1)\sigma_y^{n_{2y}+1}R_z(\phi_y)\\
        &=i\theta_{1x}\partial_y\sigma_z\sigma_yR_z(-2(\zeta_x+\zeta_y+\phi_x)),
    \end{split}
\end{equation}
and the last $\theta_{1y}\partial_x$ term:
\begin{equation}
    \begin{split}
        \sum_{\substack{l_{1y},l_{2y}\\n_{1x},n_{2x}\\=\{0,1\}}} \hat{t}_{00l_{1y}l_{2y}n_{1x}n_{2x}00}&=\partial_x\frac{i\theta_{1y}}{2}R_z(\zeta_x)\sigma_y^{n_{1x}}R_z(a_1)\sigma_z^{l_{1y}}\sigma_yR_z(a_2)\sigma_y^{n_{2x}}\\
        &\times R_z(a_1)\sigma_z^{l_{2y}}\sigma_yR_z(\phi_y)\\
        &=i\theta_{1y}\partial_x\sigma_z\sigma_yR_z(-2(\zeta_x+\zeta_y+\phi_x)).
    \end{split}
\end{equation}
These are the only non-zero terms in the continuum limit, so the full continuum limit time evolution equation for $a=b=1/2$ is the following:
\begin{equation}\label{eq:continuum_Limit_reduced}
    \begin{split}
        \partial_t\Psi(x,y,t)=(\hat{P_x}\partial_x+\hat{P_y}\partial_y)\Psi(x,y,t)
    \end{split}
\end{equation}
where
\begin{equation}
    \begin{split}
        \hat{P_x}&=i\theta_{1x}\sigma_zR_z(2(\phi_y+\zeta_x))+i\theta_{1y}\sigma_z\sigma_yR_z(-2(\zeta_x+\zeta_y+\phi_x))\\
        \hat{P_y}&=i\theta_{1y}\sigma_z\sigma_yR_z(2\phi_y)+i\theta_{1x}\sigma_z\sigma_yR_z(-2(\zeta_x+\zeta_y+\phi_x))
    \end{split}
\end{equation}
It can also easily be seen that $[\hat{P_x},\hat{P_y}]\neq 0$ for all values of $\zeta_x,~\zeta_y,~\phi_x,~\phi_y,~\theta_{1x},$ and $\theta_1y$.
%
%
%
%
%
\section{Conclusion}\label{sec:conclusion}

We introduced a QW over the 2D+1 spacetime grid, and we parametrized the walk with 9 parameters (4 for each coin, and 1 for the stroboscopic time step). We further allowed the coin parameters to be truncated Taylor polynomial at first order $\varepsilon$, which introduced 8 more free parameters. We showed that some of those parameters (the $\theta_{0i}$, $i=x,y$), must be constrained in a particular way for the continuous time and continuous spacetime limit to exist, and that the stroboscopic time step $\tau$ must be even to have both. We called this large family of QWs plastic.
We then used these constraint equations to derive a lattice Hamiltonian on a 2D-grid in continuous time (i.e. for $\Delta =1$) and a very general transport equation with dispersion terms in $2+1$ spacetime dimensions when both $\Delta_t$ and $\Delta$ tend to zero. In particular we have shown that this last PDE includes the massless Dirac Equation. 
This opens the route for elaborating QW-based quantum simulators of interacting particles admitting both non relativistic ($\Delta_t\ll\Delta$) and relativistic regime ($\Delta_t\simeq\Delta$) in a very elegant way. Moreover, the non-relativistic, naive lattice fermion Hamiltonians are known to suffer the fermion-doubling problem, i.e. a spurious degree of freedom. On the other hand, the DTQW does not suffer this problem. An intriguing question is whether the model hereby presented, suffers this problem or not. We leave this as an open question. Moreover, the methods used in this work lay the groundwork for use in a general $n$D+1 dimensional DTQW continuous time limit. 

\section{Acknowledgements}
The authors acknowledge inspiring conversations with Pablo Arrighi, Tamiro Villazon, and Pieter W. Claeys. This work has been funded by the P{\'e}pini{\`e}re d'Excellence 2018, AMIDEX fondation, project DiTiQuS and the ID \#60609 grant from the John Templeton Foundation, as part of the "The Quantum Information Structure of Spacetime (QISS)" Project.
\newpage
%
%
%
%
%
%
\begin{appendices}
%
%
%
%
%
%
\section{$W$ Expansion}\label{app:Wexpansion}
We wish to expand $\widehat{W}^\tau$ to first order in $\Delta_t$. We begin by expanding $\widehat{S_m}C_m$ up to $O(\Delta_t^2)$, where $m=x,y$:
\begin{equation}
    \begin{split}
        \widehat{S_m}C_m&=e^{i\delta_m} [R_z(\zeta'_{m})R_y(\theta_{0m})R_z(\phi_{m})-\frac{i \Delta_t}{2}(\zeta_{1m} \sigma_z R_z(\zeta'_{m})R_y(\theta_{0m})R_z(\phi_{m})\\
        &+\theta_{1m} \sigma_y R_z(-2\zeta'_{m})R_z(\zeta'_{m})R_y(\theta_{0m})R_z(\phi_{m})\\
        &+\phi_{1m} R_z(\zeta'_{m})R_y(\theta_{0m})R_z(\phi_{m}) \sigma_z)+O(\Delta_t^2)]\\
        &=e^{i\delta_m}(A_m-\frac{i \Delta_t}{2} B_m+O(\Delta_t^2))
    \end{split}
\end{equation}
Now we can combine the product of $\widehat{S_x}C_x$ and $\widehat{S_y}C_y$ up to $O(\Delta_t)$:
\begin{equation}
    \begin{split}
    \widehat{W}=\widehat{S_x}C_x\widehat{S_y}C_y&=e^{i\delta_x}(A_x-\frac{i \Delta_t}{2} B_x+O(\Delta_t^2))e^{i\delta_y}(A_y-\frac{i \Delta_t}{2} B_y+O(\Delta_t^2))\\
    &=e^{i(\delta_x+\delta_y)}(A_xA_y-\frac{i\Delta_t}{2}(A_xB_y+B_xA_y)+O(\Delta_t^2))\\
    &=e^{i\delta}(A-\frac{i\Delta_t}{2}B+O(\Delta_t^2))
    \end{split}
\end{equation}
And now we expand $\widehat{W}^\tau$ in powers of $\Delta_t$:
\begin{equation}
    \begin{split}
        \widehat{W}^\tau=(\widehat{S_x}C_x\widehat{S_y}C_y)^\tau&=e^{i\delta \tau}(A-\frac{i \Delta_t}{2} B+O(\Delta_t^2))^\tau\\
        &=e^{i\delta \tau}(A^\tau-\frac{i \Delta_t}{2}(A^{\tau-1}B+A^{\tau-2}BA+...+ABA^{\tau-2}+BA^{\tau-1})+O(\Delta_t^2))\\
        &=e^{i\delta \tau}(A^\tau-\frac{i \Delta_t}{2}\sum_{j=0}^{\tau-1}A^{\tau-1-j}BA^j+O(\Delta_t^2))\\
        &=(e^{i\delta}A)^\tau(1-\frac{i \Delta_t}{2}A^{-1}\sum_{j=0}^{\tau-1}A^{-j}BA^j+O(\Delta_t^2))
    \end{split}
\end{equation}
%
%
%
%
%
%
\newpage
\section{$\{A,B\}$ Expansion}\label{app:commutatorExpansion}
We begin by expanding $\{A,B\}$ in terms of $A_x$, $A_y$, $B_x$, and $B_y$, using $A=A_xA_y$ and $B=A_xB_y+B_xA_y$:
\begin{equation}
    \begin{split}
            \{A,B\}&=A_xA_yA_xB_y+A_xB_yA_xA_y+A_xA_yB_xA_y+B_xA_yA_xA_y.
    \end{split}
\end{equation}
Now we expand further using $B_m=\zeta_{1m}\sigma_z A_m+\theta_{1m} \sigma_y R_z(-2\zeta'_{m})A_m+\phi_{1m} A_m \sigma_z$ (where $m=x$ or $y$):
\begin{equation}\label{eq:anti-commutator}
    \begin{split}
        \{A,B\}=&\zeta_{1y}A_xA_yA_x\sigma_zA_y+\theta_{1y}A_xA_yA_x\sigma_yR_z(-2\zeta_{0y}')A_y+\phi_{1y}A_xA_yA_xA_y\sigma_z\\
        +&\zeta_{1y}A_x\sigma_zA_yA_xA_y+\theta_{1y}A_x\sigma_yR_z(-2\zeta_{0y}')A_yA_xA_y+\phi_{1y}A_xA_y\sigma_zA_xA_y\\
        +&\zeta_{1x}A_xA_y\sigma_zA_xA_y+\theta_{1x}A_xA_y\sigma_yR_z(-2\zeta_{0x}')A_xA_y+\phi_{1x}A_xA_yA_x\sigma_zA_y\\
        +&\zeta_{1x}\sigma_zA_xA_yA_xA_y+\theta_{1x}\sigma_yR_z(-2\zeta_{0x}')A_xA_yA_xA_y+\phi_{1x}A_x\sigma_zA_yA_xA_y.
    \end{split}
\end{equation}
Using $\sigma_zA_x=(-1)^\nu A_x\sigma_z$ and $\sigma_zA_y=(-1)^{\nu+1} A_y\sigma_z$, it can be shown that the first and third columns of Eq.~\eqref{eq:anti-commutator} cancel. Now we expand the remaining terms using $A_m=R_z(\zeta'_{m})R_y(\theta_{0m})R_z(\phi_{m})$:
\begin{equation}
    \begin{split}
        &\theta_{1y}A_xA_yA_x\sigma_yR_z(-2\zeta'_{0y})A_y=-\theta_{1y}R_z(-2\phi_{0y})\sigma_y\\
        &\theta_{1y}A_x\sigma_yR_z(-2\zeta'_{0y})A_yA_xA_y=-\theta_{1y}R_z(2\zeta'_{0x}+2\phi_{0x}(-1)^\nu+2\zeta'_{0y}(-1)^\nu)\sigma_y\\
        &\theta_{1x}A_xA_y\sigma_yR_z(-2\zeta'_{0x})A_xA_y=-\theta_{1x}R_z(2\zeta'_{0y}(-1)^\nu-2\phi_{0y}+2\phi_{0x}(-1)^\nu)\sigma_y\\
        &\theta_{1x}\sigma_yR_z(-2\zeta'_{0x})A_xA_yA_xA_y=-\theta_{1x}R_z(2\zeta'_{0x})\sigma_y.
    \end{split}
\end{equation}
Putting it all together, we have the following for $\{A,B\}$:
\begin{equation}
    \begin{split}
        \{A,B\}=&-\theta_{1y}(R_z(-2\phi_{0y})+R_z(2\zeta'_{0x}+2\phi_{0x}(-1)^\nu+2\zeta'_{0y}(-1)^\nu))\sigma_y\\
        &-\theta_{1x}(R_z(2\zeta'_{0x})+R_z(2\zeta'_{0y}(-1)^\nu-2\phi_{0y}+2\phi_{0x}(-1)^\nu))\sigma_y
    \end{split}
\end{equation}
%
%
%
%
%
%
\newpage
\section{Cross Term Constraint}\label{app:cross_term_constraints}
In this section, we will deduce for which constraints from section~\ref{sec:plasticQW} do the resulting continuum limit PDEs include a cross derivative term. We will find that the only constraints which will have cross terms will be the pair $\cos(\frac{\phi_{x}+\phi_{y}+\zeta_{x}+\zeta_{y}}{2})=\cos(\frac{a_1+a_2}{2})=0\rightarrow a_1+a_2=2\pi m +\pi$ and $\cos(\frac{a_1-a_2}{2})=0\rightarrow a_1-a_2=2\pi t +\pi$. We first reiterate the definitions of $\hat{\Gamma}_{l_xl_yn_xn_y}$ and $\nu'_{l_xl_yn_xn_y}$:
\begin{equation}
    \begin{split}
        &\hat{\Gamma}_{l_xl_yn_xn_y}=R_z(\zeta_{x})\sigma_z^{l_x}\sigma_y^{n_x}R_y(\theta_{0x})R_z(\phi_{x})R_z(\zeta_{y})\sigma_z^{l_y}\sigma_y^{n_y}R_y(\theta_{0y})R_z(\phi_{y})\\
        &\nu'_{l_xl_yn_xn_y}=\frac{(\partial_x)^{l_x}(\partial_y)^{l_y}(-\frac{i\theta_{1x}}{2})^{n_x}(-\frac{i\theta_{1y}}{2})^{n_y}}{l_x!l_y!n_x!n_y!}
    \end{split}
\end{equation}
Since cross terms have all $n_{vm}$s equal to zero and two $l_{vm}$s equal to one, we write the proportionality expression for $\nu'_{l_{1x}l_{1y}00}\nu'_{l_{2x}l_{2y}00}\hat{\Gamma}_{l_{1x}l_{1y}00}\hat{\Gamma}_{l_{2x}l_{2y}00}$ (where $a_1=\phi_{x}+\zeta_{y}$ and $a_2=\phi_{y}+\zeta_{x}$):
\begin{equation}
    \begin{split}
        \nu'_{l_{1x}l_{1y}00}\nu'_{l_{2x}l_{2y}00}\hat{\Gamma}_{l_{1x}l_{1y}00}\hat{\Gamma}_{l_{2x}l_{2y}00}&\propto\partial_{x}^{l_{1x}+l_{2x}}\partial_{y}^{l_{1y}+l_{2y}}\\
        &\times R_z(\zeta_{x})\sigma_z^{l_{1x}}R_y(\theta_{0x})R_z(\phi_{x})R_z(\zeta_{y})\sigma_z^{l_{2x}}R_y(\theta_{0y})R_z(\phi_{y})\\
        &\times R_z(\zeta_{x})\sigma_z^{l_{1y}}R_y(\theta_{0x})R_z(\phi_{x})R_z(\zeta_{y})\sigma_z^{l_{2y}}R_y(\theta_{0y})R_z(\phi_{y})\\
        &=\partial_{1x}^{l_{1x}+l_{2x}}\partial_{1y}^{l_{1y}+l_{2y}}\\
        &\times R_z(\zeta_{x})\sigma_z^{l_{1x}}R_y(\theta_{0x})R_z(a_1)\sigma_z^{l_{2x}}R_y(\theta_{0y})R_z(a_2)\\
        &\times \sigma_z^{l_{1y}}R_y(\theta_{0x})R_z(a_1)\sigma_z^{l_{2y}}R_y(\theta_{0y})R_z(\phi_{y})\\
        &=\partial_{1x}^{l_{1x}+l_{2x}}\partial_{1y}^{l_{1y}+l_{2y}}\sigma_z^{l_{1x}+l_{1y}+l_{2x}+l_{1y}} R_z(\zeta_{x})\\
        &\times R_y((-1)^{l_{1y}+l_{2x}+l_{1y}}\theta_{0x})R_z(a_1)R_y((-1)^{l_{2x}+l_{1y}}\theta_{0y})\\
        &\times R_z(a_2)R_y((-1)^{l_{2y}}\theta_{0x})R_z(a_1)R_y(\theta_{0y})R_z(\phi_{y})
    \end{split}
\end{equation}
Since each cross term will be proportional to $\partial_x\partial_y$, we define the following matrix to contain the relevant parts of the above equation to our analysis:
\begin{equation}
    \begin{split}
        \hat{J}_{l_{1x}l_{1y}l_{2x}l_{2y}}&=R_y((-1)^{l_{1y}+l_{2x}+l_{2y}}\theta_{0x})R_z(a_1)R_y((-1)^{l_{2x}+l_{2y}}\theta_{0y})R_z(a_2)R_y((-1)^{l_{2y}}\theta_{0x})
    \end{split}
\end{equation}
For the cross derivative terms in Eq.~\eqref{eq:csct_ham_final} to cancel, the following must be true:
\begin{equation}\label{eq:cancellation}
    \begin{split}
            &\hat{J}_{1100}+\hat{J}_{1001}+\hat{J}_{0110}+\hat{J}_{0011}=0\\
            &\rightarrow R_y(-\theta_{0x})R_z(a_1)R_y(\theta_{0y})R_z(a_2)R_y(\theta_{0x})\\
            &+R_y(-\theta_{0x})R_z(a_1)R_y(-\theta_{0y})R_z(a_2)R_y(-\theta_{0x})\\
            &+R_y(\theta_{0x})R_z(a_1)R_y(-\theta_{0y})R_z(a_2)R_y(\theta_{0x})\\
            &+R_y(\theta_{0x})R_z(a_1)R_y(\theta_{0y})R_z(a_2)R_y(-\theta_{0x})=0
    \end{split}
\end{equation}
We see that either the first two terms can cancel when the $\partial_y$  non-divergent constraint in Eq.~\eqref{eq:div_constraints_list} is imposed, or the second and fourth term can cancel when the $\partial_x$ non-divergent constraint in Eq.~\eqref{eq:div_constraints_list} is imposed. When the $\partial_y$ constraint is imposed again, the above equation reduces to the following:
\begin{equation}
    R_y(-2\theta_y)-R_y(2\theta_y)=0\to \theta_y=n\pi \text{~for n=1, 2, 3, ...}.
\end{equation}
Similarly, when the $\partial_x$ constraint is imposed again, we recover the following:
\begin{equation}
    R_y(-2\theta_x)-R_y(2\theta_x)=0\to \theta_x=m\pi \text{~for m=1, 2, 3, ...}.
\end{equation}
Thus, we see that these cross derivative terms will cancel with either $\theta_x$ or $\theta_y$ equal to an integer multiple of $\pi$. Therefore, most of the constraints will contain no cross terms. The only set of constraints which will have cross terms will be the pair $\cos(\frac{\phi_{x}+\phi_{y}+\zeta_{x}+\zeta_{y}}{2})=\cos(\frac{a_1+a_2}{2})=0\rightarrow a_1+a_2=2\pi m +\pi$ and $\cos(\frac{a_1-a_2}{2})=0\rightarrow a_1-a_2=2\pi t +\pi$, as there is no constraints on $\theta_x$ or $\theta_y$ equalling an integer multiple of $\pi$.
\end{appendices}


%
%

\bibliographystyle{spphys}       
\bibliography{biblio}   

\begin{thebibliography}{10}
\providecommand{\url}[1]{{#1}}
\providecommand{\urlprefix}{URL }
\expandafter\ifx\csname urlstyle\endcsname\relax
  \providecommand{\doi}[1]{DOI \discretionary{}{}{}#1}\else
  \providecommand{\doi}{DOI \discretionary{}{}{}\begingroup
  \urlstyle{rm}\Url}\fi

\bibitem{molfetta2019quantum}
G.~Di~Molfetta, P.~Arrighi, Quantum Information Processing \textbf{19}(2), 47
  (2020)

\bibitem{FeynmanQC}
R.P. Feynman, International Journal of Theoretical Physics \textbf{21}(6), 467
  (1982)

\bibitem{georgescu2014quantum}
I.M. Georgescu, S.~Ashhab, F.~Nori, Reviews of Modern Physics \textbf{86}(1),
  153 (2014)

\bibitem{jordan2012quantum}
S.P. Jordan, K.S. Lee, J.~Preskill, Science \textbf{336}(6085), 1130 (2012)

\bibitem{StrauchCTQW}
F.W. Strauch, Physical Review A \textbf{74}(3), 030301 (2006)

\bibitem{Arnault_2017}
P.~Arnault, F.~Debbasch, Annals of Physics \textbf{383}, 645–661 (2017).
\newblock \doi{10.1016/j.aop.2017.04.003}.
\newblock \urlprefix\url{http://dx.doi.org/10.1016/j.aop.2017.04.003}

\bibitem{di2012discrete}
G.~Di~Molfetta, F.~Debbasch, Journal of Mathematical Physics \textbf{53}(12),
  123302 (2012)

\bibitem{MolfettaDebbasch2014Curved}
G.~Di~Molfetta, M.~Brachet, F.~Debbasch, Physica A: Statistical Mechanics and
  its Applications \textbf{397}, 157 (2014)

\bibitem{Arnault_2016}
P.~Arnault, F.~Debbasch, Physical Review A \textbf{93}(5) (2016).
\newblock \doi{10.1103/physreva.93.052301}.
\newblock \urlprefix\url{http://dx.doi.org/10.1103/PhysRevA.93.052301}

\bibitem{di2016quantum}
G.~Di~Molfetta, A.~P{\'e}rez, New Journal of Physics \textbf{18}(10), 103038
  (2016)

\bibitem{arnault2016quantum}
P.~Arnault, G.~Di~Molfetta, M.~Brachet, F.~Debbasch, Physical Review A
  \textbf{94}(1), 012335 (2016)

\bibitem{ARNAULT2016179}
P.~Arnault, F.~Debbasch, Physica A: Statistical Mechanics and its Applications
  \textbf{443}, 179  (2016).
\newblock \doi{https://doi.org/10.1016/j.physa.2015.08.011}.
\newblock
  \urlprefix\url{http://www.sciencedirect.com/science/article/pii/S0378437115006664}

\bibitem{di2013quantum}
G.~Di~Molfetta, M.~Brachet, F.~Debbasch, Physical Review A \textbf{88}(4),
  042301 (2013)

\bibitem{ArrighiGRDirac3D}
P.~Arrighi, F.~Facchini, Quantum Information and Computation \textbf{17}(9-10),
  0810 (2017).
\newblock \urlprefix\url{https://arxiv.org/abs/1609.00305}.
\newblock ArXiv:1609.00305

\bibitem{succiQWBoltzmann}
S.~Succi, F.~Fillion-Gourdeau, S.~Palpacelli, EPJ Quantum Technology \textbf{2}
  (2015).
\newblock \doi{10.1140/epjqt/s40507-015-0025-1}

\bibitem{arrighi2019curved}
P.~Arrighi, G.~Di~Molfetta, I.~M{\'a}rquez-Mart{\'\i}n, A.~P{\'e}rez,
  Scientific reports \textbf{9}(1), 1 (2019)

\bibitem{kogut1975hamiltonian}
J.~Kogut, L.~Susskind, Physical Review D \textbf{11}(2), 395 (1975)

\bibitem{zohar2015formulation}
E.~Zohar, M.~Burrello, Physical Review D \textbf{91}(5), 054506 (2015)

\bibitem{manighalam}
M.~Manighalam, M.~Kon.
\newblock Continuum limits of the 1d discrete time quantum walk (2019)

\end{thebibliography}

%
%

\end{document}